\documentclass[onecolumn,a4paper,11pt,
]{quantumarticle}
\pdfoutput=1

\usepackage{amsmath}
\usepackage{amsfonts}
\usepackage{amssymb}
\usepackage{amsthm}
\usepackage{mathtools}
\usepackage{mathrsfs}
\usepackage{enumitem}

\usepackage{uniinput}
\usepackage[utf8]{inputenc}
\usepackage[english]{babel}
\usepackage[T1]{fontenc}
\usepackage{xcolor}
\usepackage{hyperref}
\hypersetup{colorlinks = true}

\usepackage{braket}
\usepackage{amsthm}

\usepackage[numbers,compress]{natbib}
\theoremstyle{plain}
\newtheorem{thm}[]{Theorem}
\newtheorem{lem}[thm]{Lemma}
\newtheorem{cor}[thm]{Corollary}
\theoremstyle{definition}
\newtheorem{defn}[thm]{Definition}

\begin{document}
	
	\title{Dequantizability from inputs}
	\author{Tae-Won Kim}
	\affiliation{Department of Computer Engineering and Artifical Intelligence, Pukyong National University, Busan 48513, South Korea}
	\author{Byung-Soo Choi}
	\email{Corresponding Author: bschoi@pknu.ac.kr}
	\homepage{https://sites.google.com/view/bschoi/}
	\affiliation{Department of Scientific Computing, Pukyong National University, Busan 48513, South Korea}
	\maketitle
	
	\begin{abstract}
		By comparing the constructions of block encoding given by \cite{kp17b,kp17a,csg19,gslw19}, we propose a way to extract \emph{dequantizability} from advancements in dequantization techniques that have been led by Tang, as in \cite{tang23}. Then we apply this notion to the sparse-access input model that is known to be \textsf{BQP}-complete in general, thereby conceived to be un-dequantizable. Our goal is to break down this belief by examining the sparse-access input model's instances, particularly their input matrices. In conclusion, this paper forms a dequantizability-verifying scheme that can be applied whenever an input is given. 
	\end{abstract}
	
	\section{Introduction}
	Quantum machine learning (QML) explores whether principles of quantum mechanics would improve machine learning, seeking the practical, real-world utility of quantum computers. QML was founded by the quantum linear system solver introduced in \cite{hhl09}, which demonstrated significant speedups, particularly exponential ones, when the input matrix is sparse, well-conditioned, and efficiently accessible by quantum algorithms using the bucket-brigade qRAM as in \cite{glm08}. Despite the skepticism noted in \cite{aar15} that was led by these numerous caveats and fine details associated, HHL appeared promising as a subroutine for future QML algorithms, prompting an atomic subroutine that unifies various QML algorithms and their source of speedup, as in \cite{lc19}. To this end, from \emph{block encoding}, constructed with regard to quantum access by \cite{csg19}, a framework of quantum singular value transformation (QSVT) was proposed by \cite{gslw19}, that may ``grand-unify'' quantum algorithms. 
	
	Meanwhile, with respect to quantum access that was justified by the data structure of \cite{kp17b, kp17a}, polynomially slower classical counterparts of QML algorithms were provided by sample-query access corresponding to the data structure, namely, \emph{dequantization}: eventually dequantizing the QSVT framework in \cite{bt23}. Results of dequantization are depicted in \cite[Figure 7]{tang23}, where it appears most QML algorithms have been dequantized, including the QSVT framework, suggesting that QML in general may not provide exponential speedup at all, quadratic being the best-case. 
	
	However, there are limitations to dequantization. Note that all dequantized QML algorithms only take input matrices that get normalized by the Frobenius norm to be prepared for a quantum state. So these dequantized QML algorithms are the ones that utilize a specific class of instances from quantum-accessible data structures within qRAM. Hence, the performance comparison between quantum and classical machine learning is reduced to comparing the construction cost of the data structure, as noted in \cite{jr23}. Also, since qRAM is not the only way to provide quantum-access, dequantization cannot capture query-access instances via circuits with some polylogarithmic depth, as in \cite{zly22}.  
	
	More crucially, we do not know whether the complement class of Frobenius-normalising instances of the quantum-accessible data structure would allow dequantization. Meaning that ``full'' QSVT is conceived to be un-dequantizable because it includes \emph{sparse-access input model}, as in the sparse linear systems solver of \cite{hhl09}, which is known to be \textsf{BQP}-complete. Therefore, in \cite[p.21]{tang23}, Tang posed the open problem of an alternative construction of block encoding or state preparation that may prevent dequantizaition, where such a complement class is mentioned as an `alternative data structure.' 
	
	\subsection{Main Ideas}
	Section \ref{sec2} questions whether such a component class is an \emph{independent} alternative data structure rather than one of the two inner models of the data structure. We state that if we assume that matrices are provided separately in separate data structures in the qRAM, a contradiction follows, which would be our primary result against the consensus in the construction of block encoding. We provided the following lemma, a negation of \cite[Corollary 5]{csg19}.
	
	\begin{lem}[Informal, Lemma \ref{transdisproof}]
		There is a matrix $A\in\mathbb{C}^{m\times n}$ stored in a quantum-accessible data structure such that no quantum algorithm
		performs $s_p$-normalizing state preparation of $A$ in $\textrm{\normalfont{polylog}}(mn/\epsilon)$ time.
	\end{lem}
	
	Meaning that, despite some matrix being stored in the quantum-accessible data structure, the storage alone does not imply the existence of certain efficient quantum access to that matrix. Instead, some property of the matrix is required for certain efficient quantum access. In other words, let any matrix $A$ be stored in qRAM. Then, the pass towards $A$ would determine one of the two inner models. Hence, $A$ may be prepared as a quantum state via Frobenius-normalisation or via another way defined in \cite{kp17a}. We highlight that since the input matrix $A$ and its properties such as sparsity are the conditions that initiate Frobenius-normalisation quantum-accessible data structure, which is the space where current dequantizations lie, we can simplify dequantization results as some bound for the input matrix $A$. We construct such a sparsity bound and justify it as a notion of \emph{dequantizability from matrices} in Section 3, where we aim to provide a statement like this:
	
	\begin{lem}[Informal, Lemma \ref{deq}]
		If $A$ is a matrix that satisfies the sparsity bound $\mathscr{S}$, then there exists some quantum algorithms $\mathcal{Q}:A\mapsto\alpha$ such that some classical algorithm dequantizes $\mathcal{Q}$. 
	\end{lem}
	
	Then we take the contrapositive of this contingently trivial fact to obtain a statement general enough that we may ignore whether the algorithm is provided by query-access via polynomial depth circuits or by quantum-access via qRAM data structure.
	
	\begin{lem}[Informal, Lemma \ref{lem:undeq}]
		For any quantum algorithms $\mathcal{Q}:A\mapsto\alpha$, if no classical algorithm dequantizes $\mathcal{Q}$, then $A$ is a matrix such that satisfies the negation of $\mathscr{S}$.
	\end{lem}
	
	In Section \ref{sec:mainresults}, we let $A$ be $\kappa$-conditioned, $s$-sparse, and Hermitian. In other words, we let the instance of a quantum algorithm be in the \emph{sparse-access input model}. By assuming some instance $(A,\alpha)\in\mathcal{Q}$ that is parameterized by $\kappa$ to be un-dequantizable, the Lemma above gets utilised as Theorem \ref{thm:main}. Our main result would be its corollary below.
	
	\begin{cor}[Informal, Corollary \ref{cor:mainresult}]
		For some $1$-sparse Hermitian matrix $A\in\mathbb{C}^{n\times n}$, such as
		\[
		A = \begin{pmatrix}
			(n+\epsilon)^{-1} & 0 & \cdots &0 & 0 \\
			0& (n+\epsilon)^{-2} & \cdots &0 &0 \\
			\vdots &  & \ddots & & \vdots \\
			0 & 0 & 0 & (n+\epsilon)^{-n+1} &0 \\
			0 & 0 & 0 &0 &(n+\epsilon)^{-n}
		\end{pmatrix}
		\]
		where $\epsilon>0$, no quantum algorithm from $A$ prevents dequantization.
	\end{cor}
	
	This corollary has various implications. Primarily, it necessitates that there exists dequantization for some non-trivial sparse-access input model, so we may unexpectively generalise any particular results of dequantization. Concurrently, the whole scheme may be considered as a single proof of negation; the `quantum-accessible data structure-block encoding-QSVT' triad would not ``grand-unify'' QML algorithms at all, unless we want to permit \textsf{BQP}-completness' shrinkage by the classical data and not by the quantum algorithm. That is, the space of all possible QML algorithms or the correspondences between classical data and its quantum access is largely unknown, and deep down in such a space, sparsity would be only a fraction of the quantum-advantageous figure. In Section \ref{sec:conclusion}, we conclude with an open problem and an outlook that would limit or clarify these implications.
	
	\section{Inner models of quantum-access}\label{sec2}
	The vital component of our proof is the space where current results in dequantization lie. We excavate this from qRAM with a data structure for its quantum access. Specifically, we focus on the construction of a method for some quantum device to uphold any data stored in qRAM. That is, block encoding via a quantum-accessible data structure. 
	
	The construction of block encoding is outlined in \cite[Lemma 6]{csg19} without accompanying proof, as it is considered a direct translation of \cite[Theorem 4.4]{kp17a}. We restate the initial theorem and its translations as a setup. Note that $\|A\|_F$ denotes the Frobenius norm $\sqrt{\sum_{i=1}^m\sum_{j=1}^n|A_{ij}|^2}$ and $s_p(A)$ is (intended to be) defined as the $p$-th power of $\max_{i\in[m]} \ell_p$ norm of $A$. For $p\in[0,2]$:
	\[
	s_p(A) := \max_{i\in[m]}\|A_i\|_p^p = 
	\left(\max_{i\in[m]}\left(\sum_{j=1}^n|A_{ij}|^p\right)^{1/p}\right)^p\,.
	\]
	Note that $A^{(p)}$ is a matrix such that  $A^{(p)}_{ij} = (A_{ij})^p$ for $A\in\mathbb{C}^{m\times n}$.
	\begin{thm}[{\cite[Theorem 4.4]{kp17a}}]\label{thm:innermodel}
		Let $A\in\mathbb{C}^{m\times n}$ be stored in a quantum-accessible data structure. There is a quantum algorithm that performs singular value estimation (SVE)
		for $A$ to precision $\delta$  in time $\tilde{O}(\mu(A)/\delta)$ where
		\begin{equation}\label{mupass}
			\mu(A) = \min_{p\in[0,2]}\left(\|A\|_F, \sqrt{s_{p}(A)s_{2-p}(A^{\dagger})}\right)\,.
		\end{equation}
	\end{thm}
	\begin{lem}[{\cite[Lemma 6]{csg19}, \cite[Lemma 50]{gslw19}}]\label{lem:translation}
		Let $A\in\mathbb{C}^{m\times n}$. 
		\begin{itemize}
			\item Fix $p\in[0,2]$. If
			$A^{(p)}$ and $\left(A^{(2-p)}\right)^{\dagger}$ are both stored in
			quantum-accessible data structures, then there exist unitary matrices $U_R$ and $U_L$ that can be implemented in time
			$O(\textrm{\normalfont polylog}(mn/\epsilon))$ such that $U^{\dagger}_RU_L$ is a $\left(\mu_p(A), \lceil\log(n+m+1)\rceil,\epsilon\right)$-block encoding of $A$.
			\item If $A$ is stored in a quantum-accessible data structure, then there exists $U_R$ and $U_L$ that can be implemented in time $O(\textrm{\normalfont polylog}(mn)/\epsilon)$ such that $U^{\dagger}_RU_L$ is a $(\|A\|_F,\lceil\log(m+n)\rceil,\epsilon)$-block encoding of $A$.
		\end{itemize}
	\end{lem}
	
	The translation seems straightforward. Matrix $A$ or $A^{(p)}$ (with $(A^{(2-p)})^{\dagger}$) are stored in qRAM via the Frobenius-normalising $\mu_F$ inner model or the $s_p$-normalising $\mu_p$ inner model, respectively. These inner models enable quantum-access unitary matrices $U_R^{\dagger}$ and $U_L$ (in an SVE-accelerative manner), which form a block encoding of $A$ by the linear algebraic definition of block encoding.
	
	However, this translation is disputable. The matrix is not given inside the data structure in the first place. It is given outside, and a pass (\ref{mupass}) is initiated towards the external matrix to determine which inner model would enable the optimal normalisation factor. Then the matrix gets stored in qRAM by that optimal inner model, resulting in some matrix in qRAM. We state our point and highlight the significance of pass (\ref{mupass}) by disproving \cite[Corollary 5]{csg19}, which was devised to justify the translation.
	
	\begin{lem}[negation of {\cite[Corollary 5]{csg19}}]\label{transdisproof}
		There is a matrix $A^{(p/2)}\in\mathbb{C}^{m\times n}$ stored in a quantum-accessible data structure, such 
		that no quantum algorithm performs the map (\ref{maptransdisproof}) with 
		$\epsilon$-precision in $\textrm{\normalfont{polylog}}(mn/\epsilon)$ time.
		\begin{equation}\label{maptransdisproof}
			\ket{i}\ket{0}\mapsto\ket{i}\frac{1}{s_{p}(A)}\sum_{j=1}^{n}(A_{ij})^{p/2}\ket{j}
		\end{equation}
	\end{lem}
	\begin{proof}
		We first clarify \cite[Corollary 5]{csg19}. By Theorem \ref{thm:innermodel}, two cases are implied for assuming that $A^{(p/2)}$ is stored in a quantum-accessible data structure:
		\begin{enumerate}[label=(\roman*)]
			\item $\mu_p$ inner model is initiated towards $A$, letting $A^{(p/2)}$ in qRAM, or
			\item $\mu_F$ inner model is initiated towards $A^{(p/2)}$, letting $A^{(p/2)}$ in qRAM.
		\end{enumerate}
		Let \cite[Corollary 5]{csg19} hold, and assume that no quantum algorithm performs the map (\ref{maptransdisproof}). Then it must follow that $A^{(p/2)}$ is not in qRAM for both inner models. But since no quantum algorithm performs the map (\ref{maptransdisproof}), it is not the case that $\mu_p$ inner model is initated towards $A$, so that by Theorem \ref{thm:innermodel}, for all $p\in[0,2]$:
		\begin{equation}\label{deqineq}
			\|A\|_F \leq \sqrt{s_p(A)s_{2-p}(A^{\dagger})}\,,
		\end{equation}
		which implies
		\[
		\min_{p\in[0,2]}\left(\|A^{(p/2)}\|_F, \sqrt{s_p(A^{(p/2)})s_{2-p}(A^{(p/2)\dagger})}\right) =
		\|A^{(p/2)}\|_F\,.
		\]
		Meaning that, even if no quantum algorithm performs a map (\ref{maptransdisproof}), there can be $A^{(p/2)}$ in qRAM (via the $\mu_F$ inner model of $A^{(p/2)}$), which is a contradiction. By proof of negation, \cite[Corollary 5]{csg19} is false.
	\end{proof}
	Such contradiction could be led from Lemma \ref{lem:translation} in a similar way. This does not mean that the whole construction of block encoding from quantum-accessible data structures is invalid, nor is the whole QSVT framework based on it invalid. Our statement is that we cannot assume a matrix to be given in some quantum-accessible data structure and deduce some statement regarding block encoding, meaning that matrix could be provided outside the data structure only, where the pass (\ref{mupass}) determines its data structure afterwards. Note that the proof provided later by Chakraborty et al. in \cite[Lemma 3]{cmp23} reinforces our point rather than resolving it because their proof relies on the assumption that $U_R^{\dagger}$ and $U_L$ exist.  
	
	\section{Dequantizability}\label{sec:deq}
	We keep our focus on the pass (\ref{mupass}) and reduce inequality (\ref{deqineq}) to some sparsity bound. Beforehand, it is worth consulting a functional analytic fact, as in \cite{day40}, that $\|x\|_p:=(\sum_{i=1}^n|x_i|^p)^{1/p}$ for $p\in[0,1)$ does not define a norm because it is not sub-additive. Therefore, as in \cite{kal86}, we say \emph{quasinorm} instead, and define otherwise for $p=0$. This is not only from the fact that $\max_{i\in[n]}(\sum_{i=1}^n|x_i|^0)^{1/0}$yields a meaningless result, but also due to the $s_0$ function's primer usage, such as in proximal operators \cite{flb16}: $s_0(A)$ as the maximum count of nonzero entries per row $A_i$ for $i\in[m]$, that is exactly the definition of the sparsity $s$ in the sparse-access input model:
	
	\begin{defn}
		Matrix $A^{n\times n}$ is said to be \emph{$s$-sparse} if it has at most $s$ nonzero entries
		per row. If $s=\textrm{polylog}(n)$, $A$ is said to be \emph{sparse}.
	\end{defn}
	
	We define the statement of some quantum algorithm to be dequantized as follows:
	
	\begin{defn}[{\cite[Definition 2]{chm21}}]
		Let $\mathcal{Q}:A^{m\times n}\mapsto \alpha$ be a quantum algorithm where $\mathcal{Q}$ might
		output some state $\ket{a}$ instead of value $\alpha$. We say that classical
		algorithm $\mathcal{C}$ is a dequantization of $\mathcal{Q}$, if $\mathcal{C}$
		evaluates queries to $\ket{a}$ or ouputs $a$ with
		\begin{enumerate}
			\item a similar performance guarantee as $\mathcal{Q}$,
			\item a number of sample-query access polynomial in $m$,
			\item and a worst runtime that is polynomially slower than of $\mathcal{Q}$,
		\end{enumerate}
		from sample-query access on $A$.
	\end{defn}
	
	The expression ``similar performance guarantee'' or the term sample-query access are not defined, since (our argument is that) all we need is a single fact: existing dequantizations are provided when the $\mu_F$ inner model is initiated.
	
	\begin{lem}\label{deq}
		If $A$ is a matrix such that
		\[
		s_{2-p}\left(A^{\dagger}\right) \geq \frac{\|A\|_F^2}{s_p(A)}
		\]
		for all $p\in[0,2]$, then there exists some quantum algorithm 
		$\mathcal{Q}:A\mapsto\alpha$ such that some classical algorithm is a 
		dequantization of $\mathcal{Q}$.
	\end{lem}
	\begin{proof}
		First, we simplify the fact that existing dequantizations are provided when the $\mu_F$ inner model is initiated: there exists some quantum algorithm $\mathcal{Q}:A\mapsto\alpha$ such that some classical algorithm is a dequantization of $\mathcal{Q}$ when inequality (\ref{deqineq}) holds for $A$ and all $p\in[0,2]$. We negate inequality (\ref{deqineq}) to obtain
		\[
		\sqrt{s_p(A)s_{2-p}\left(A^{\dagger}\right)} < \|A\|_F\,,
		\]
		so that
		\[
		s_{2-p}(A^{\dagger}) < \frac{\|A\|_F^2}{s_{p}(A)}\,,
		\]
		for some $p\in[0,2]$. Therefore, if
		\[
		s_{2-p}(A^{\dagger}) \geq \frac{\|A\|_F^2}{s_p(A)}
		\]
		for all $p\in[0,2]$, inequality (\ref{deqineq}) follows, enabling some quantum algorithm $\mathcal{Q}:A\mapsto\alpha$ such
		that some classical algorithm is a dequantization of $\mathcal{Q}$. 
	\end{proof}
	\begin{lem}\label{lem:undeq}
		For any quantum algorithm $\mathcal{Q}:A\mapsto\alpha$, if no classical algorithm
		is a dequantization of $\mathcal{Q}$, then $A$ is a $s_0(A)$-sparse matrix
		such that
		\[
		s_0(A) < \frac{\|A\|_F^2}{s_2(A^{\dagger})}
		\]
		or $A^{\dagger}$ is a $s_0(A^{\dagger})$-sparse matrix where
		$s_0(A^{\dagger})$ is less than $\|A\|_F^2/s_2(A)$.
	\end{lem}
	\begin{proof}
		Take the contrapositive of Lemma \ref{deq} with a selection of $p=2$ or $p=0$.
	\end{proof}
	This seemingly obvious sparsity bound is in fact full of intricacies, in particular when $A$ is a sparse, well-conditioned Hermitian. 
	
	\section{Main results}\label{sec:mainresults}
	Sparse-access input model is commonly conceived, as in \cite{csg19, gslw19, scc24}, as a model of accessing sparse matrix $A$ that gets queried by its nonzero entries per row. But as we mentioned in the Introduction, when one says the \textsf{BQP}-completeness of sparse input model, as in \cite{tang23}, it means the \emph{fine prints} of optimal HHL for the matrix inversion problem, which is known to be \textsf{BQP}-complete. One of the fine prints is that the matrix is assumed to be well-conditioned, as specified by the condition number $\kappa$ as follows:
	\begin{defn}
		The condition number $\kappa$ of matrix $A$ is the ratio of the largest to the smallest eigenvalue of $A$, meaning that
		\[
		\kappa = \frac{\lambda_{\max}}{\lambda_{\min}}\,.
		\]
	\end{defn}
	\begin{defn}[\cite{hhl09, dhm18}]
		A quantum algorithm $\mathcal{Q}:A\mapsto\alpha$ is said to be \emph{sparse-access input} if
		$A^{n\times n}$ is a sparse Hermitian conditioned by $\kappa$
		such that
		\[
		\kappa^{-2}I \preceq A^{\dagger}A \preceq I\,,
		\]
		equivalently,
		\[
		\kappa^{-1}\leq\lambda_i\leq1
		\]
		for all $i\in[n]$.
	\end{defn}
	
	One may raise a question from Lemma \ref{lem:undeq} that, since we have separatly defined $s_p$ for $p=0$, should not the case of $p=2$ also be separatly defined? However, our focus is on the algorithms in the sparse-access input model, where for a $\kappa$-conditioned Hermitian $A$, by Lemma \ref{lem:genHol} in appendix \ref{app:genHol}
	\[
	s_2(A) \leq s_{2-\epsilon}(A)\leq s_1(A) \leq s_0(A) = s_0(A^{\dagger}) = s\,,
	\]
	so if we let $p=1$ from Lemma \ref{deq}
	\[
	s_1(A^{\dagger}) \geq \frac{\|A\|_F^2}{s_1(A)}\,, 
	\]
	then
	\[
	s(A^{\dagger}) \geq \frac{\|A\|_F^2}{s_1(A)}\,.
	\]
	So when we take its contrapositive as in Lemma \ref{lem:undeq}, by letting some quantum algorithm in sparse input model to be un-dequantizable, then $A$ is a matrix such that
	\[
	s(A^{\dagger}) < \frac{\|A\|_F^2}{s_1(A)}\,.
	\]
	Hence
	\[
	s(A^{\dagger}) < \frac{\|A\|_F^2}{s_2(A)}\,,
	\]
	which is a conclusion desired, independent from the way we define $s_0(A)$.
	\begin{thm}\label{thm:main}
		For any sparse-access input quantum algorithm $\mathcal{Q}_s:A\mapsto\alpha$, if no classical algorithm is a dequantization of
		$\mathcal{Q}_s$, then $A$ is a $\kappa$-conditioned $s$-sparse Hermitian such that
		\[
		\kappa < \frac{\sum_{i=1}^n|\lambda_i|}{\sqrt{s}|\lambda_{\min}|}\,,
		\]
		equivalently,
		\[
		\frac{\sqrt{s}}{\kappa(n-1)+1} < |\lambda_{\min}|\,.
		\]
	\end{thm}
	\begin{proof}
		Let $\mathcal{Q}_s:A\mapsto\alpha$ be a sparse-access input quantum algorithm such that no classical algorithm is a dequantization of $\mathcal{Q}_s$. By Lemma \ref{lem:undeq}, it follows that
		\[
		s(A^{\dagger}) < \frac{\|A\|_F^2}{s_2(A)}\,.
		\]
		Since $Q_s:A\mapsto\alpha$ is sparse-acesss input, it follows that:
		\[
		s(A^{\dagger}) < \frac{\sum_{i=1}^n|\lambda_i|^2}{|\lambda_{\max}|^2}\,,
		\]
		so
		\begin{align*}
			\sqrt{s} &< \frac{\left(\sum_{i=1}^n|\lambda_i|^2\right)^{1/2}}{|\lambda_{\max}|} \\
			&\leq \frac{\sum_{i=1}^n|\lambda_i|}{|\lambda_{\max}|} \textrm{\quad\quad\quad\quad Cauchy-Schwartz} \\
			&=\frac{\sum_{i=1}^n|\lambda_i|}{\kappa|\lambda_{\min}|}\,.
		\end{align*}
		Then 
		\begin{equation}\label{eq:kap}
			\begin{aligned}
				\kappa &< \frac{\sum_{i=1}^n|\lambda_i|}{\sqrt{s}|\lambda_{\min}|} \\
				&\leq \frac{[\kappa(n-1)+1]|\lambda_{\min}|}{\sqrt{s}|\lambda_{\min}|} \\
				&= \frac{\kappa(n-1)+1}{\sqrt{s}}\,,
			\end{aligned}
		\end{equation}
		equivalently,
		\begin{align*}
			|\lambda_{\min}| &> \frac{\sqrt{s}}{\kappa(n-1)+1}\,.
		\end{align*}
	\end{proof}
	Note that the purpose of Theorem \ref{thm:main} is to showcase the method of dequantizability-verification from input, beforehand algorithm analysis or execution. For example, if $A$ is an identity matrix, or, in other words, if $A$ is a $1$-sparse, $1$-conditioned Hermitian, then
	\[
	\frac{1}{n} < 1 < n,
	\]
	which is trivial. Nevertheless, we may apply Theorem \ref{thm:main} to state that a $1$-sparse Hermitian of an arbitrary size would lead to a quantum algorithm that cannot prevent dequantization. In other words, sparsity alone is not enough for QML algorithms to be quantum advantageous.
	\begin{cor}\label{cor:mainresult}
		Let a $1$-sparse, $n\times n$ Hermitian $A$ with diagonal entries $d^{-i}$ for all $i\in[n]$ and $d>n$. Then no quantum algorithm
		$\mathcal{Q}:A\mapsto\alpha$ prevents dequantization.
	\end{cor}
	\begin{proof}
		Let $A$ be a $1$-sparse, $n\times n$ Hermitian with diagonal entries $d^{-i}$ for all $i\in[n]$ and $d>n$. For an example, with
		$n=3$ and $d=4$, matrix $A$ would be
		\[
		\begin{pmatrix}
			\frac{1}{4^1} & 0 & 0\\
			0 & \frac{1}{4^2} & 0 \\
			0 & 0& \frac{1}{4^3}
		\end{pmatrix}\,.
		\]
		Anyhow, note that
		\[
		\kappa=\frac{\lambda_{\max}}{\lambda_{\min}}=\frac{1/d^1}{1/d^n}=d^{n-1}\,,
		\]
		which is greater or equal to $1$ for any $n$. Hence any quantum algorithm from $A$ would be
		sparse-access input. We assumme that there is a quantum algorithm $\mathcal{Q}:A\mapsto\alpha$ such that
		no classical algorithm dequantizes $\mathcal{Q}$. By Theorem \ref{thm:main}, it follows that
		\[
		\lambda_{\min} > \frac{\sqrt{s}}{\kappa(n-1)+1}\,,
		\]
		so
		\[
		\frac{1}{d^n} > \frac{1}{d^{n-1}(n-1)+1}\,.
		\]
		Then
		\[
		d^n < d^{n-1}(n-1)+1\,,
		\]
		for which
		\[
		d < (n-1)+\frac{1}{d^{n-1}}\,,
		\]
		where we assumed that $n<d$. Therefore $1<d-n+1$ with
		\[
		1 < \frac{1}{d^{n-1}}\,,
		\]
		which is a contradiction for any $d>n$. By proof of negation, for any sparse-access quantum algorithm $\mathcal{Q}:A\mapsto\alpha$,
		there is a classical algorithm such that dequantizes $\mathcal{Q}$.
	\end{proof}
	Note that if the assumption that $\lambda_{\max}=1$ is necessary, as in the linear combinations of unitary (LCU) approach of \cite[Theorem 3, Lemma 11]{cks17}, we cannot lead to a contradiction for a $1$-sparse Hermitian, at least for now, thereby requiring some optimization. For instance, if $s>1$, we may rather consider pre-structures of the matrix, regarding the trace of $\lambda$ in inequality (\ref{eq:kap}), which may bridge research on the direct construction of a (pre-structured) matrix to block encoding, as in \cite{scc24}, and research on quantum and quantum-inspired machine learning.
	
	\section{Conclusions and Outlook}\label{sec:conclusion}
	The lower bound of the eigenvalue transformation, and thus of QSVT, only requires the matrix for block-encoding to be an unknown Hermitian, with the promise that the spectrum lies in $I$, as in \cite[Theorem 73]{gslw19}. This means that proving or disproving the statement---no classical algorithm dequantizes $\mathcal{Q}:A\mapsto\alpha$ for any unknown $\kappa$-conditioned $s$-sparse Hermitian $A$---will either limit or clarify the implications of Corollary \ref{cor:mainresult}, mentioned earlier.   
	
	 So, how could we articulate the inequality in Theorem \ref{thm:main}? To address this, that is, to directly exploit the $s_p$ function or the $L^p$ space in terms of quantum access of classical data, we may effectively reformulate the problem as a matter of signal processing, in particular sparcification or proximal operators, as mentioned briefly in Section \ref{sec:deq}. This approach towards QML algorithm design and analysis would bring us back to its foundation, or the equation $Ax=b$.  
		
	\section*{Acknowledgements}
	This work is supported in part by the National Research Foundation of Korea (NRF) grant funded by the Korean Government Ministry of Science and ICT (MSIT) under Grant 2020K1A3A1A78087782, in part by the Pukyong National University Industry-university Cooperation Research Fund in 2023 (202312370001).
	  
	\bibliography{dequantizability_quantjournal}
	
	\appendix
	\section{Ordering on $s_p$: H\"older inequality}\label{app:genHol}
	Our main result exploited the definition of $s_p$ function on $p\in[0,2]$, which we may justify by the Generalized H\"older inequality. Note that Theorem \ref{thm:genHol} and its Corollary are detailed in some standard textbooks on matrix analysis such as \cite[Chapter 5]{bha97}. 
	\begin{thm}[Generalized H\"older inequality, {\cite[Section 5.2]{bha97}}]
		\label{thm:genHol}
		Let $v$ and $w$ be vectors in $\mathbb{C}^n$. For $p,q$ such that $1/p+1/q=1$, it follows that
		\[
		\sum_{i=1}^n|v_iw_i|\leq\left(\sum_{i=1}^n|v_i|^p\right)^{1/p}\left(\sum_{i=1}^n|w_i|^q\right)^{1/q}\,.
		\]
	\end{thm}
	\begin{cor}\label{cor:genHol}
		Let $x$ be a vector in $\mathbb{C}^n$. For any $p_1,p_2$ such that $p_1\leq p_2$, it follows that
		\[
		\left(\sum_{i=1}^n|x_i|^{p_2}\right)^{1/p_2}\leq\left(\sum_{i=1}^n|x_i|^{p_1}\right)^{1/p_1}\,.
		\] 
	\end{cor}
	\begin{proof}
		From Theorem \ref{thm:genHol}, let $w_i=1$ to obtain:
		\[
		\sum_{i=1}^n|v_i|\leq\left(\sum_{i=1}^n|v_i|^p\right)^{1/p}\cdot n^{1/q}\,.
		\]
		To ignore $n^{1/q}$, Let $1/q<0$, thereby $p\in(0,1]$. Then it follows that
		\[
		\left(\sum_{i=1}^n|v_i|\right)^p\leq\sum_{i=1}^n|v_i|^p\,.
		\]
		By assumption, $(p_1/p_2)\leq 1$, so for $|x_i|^{p_2}:=|v_i|$
		\[
		\left(\sum_{i=1}^n|x_i|^{p_2}\right)^{p_1/p_2}\leq\sum_{i=1}A^n|x_i|^{p_2(p_1/p_2)}\,.
		\]
		Therefore,
		\[
		\left(\sum_{i=1}^n|x_i|^{p_2}\right)^{(1/p_1)(p_1/p_2)}\leq\left(\sum_{i=1}^n|x_i|^{p_2(p_1/p_2)}\right)^{(1/p_1)}\,,
		\]
		where we may discard the $(1/p_1)$ from the right hand. Since $(1/p_1)(p_1/p_2)=1/p_2$ and $p_2(p_1/p_2)=p_1$, the desired inequality follows. 
	\end{proof}
	As we've discussed in Section \ref{sec2}, $A$ is initially provided outside the data structure, so the $\max$ part from the $s_p(A)$ function does not matter much.
	\begin{lem}\label{lem:genHol}
		Let a $\kappa$-conditioned Hermitian $A$ such that $k^{-1}\leq \lambda_i\leq 1$. It follows that
		\[
		s_2(A)\leq s_{2-\epsilon}(A)\leq s_1(A)\leq s_0(A),
		\]
		where $\epsilon\in(0,1)$.
	\end{lem}
	\begin{proof}
		By Corollary \ref{cor:genHol}, for all $i\in[m]$
		\[
		\max_{i\in[n]}\left(\sum_{j=1}^n|A_{ij}|^2\right)^{1/2}\leq\max_{i\in[n]}\left(\sum_{j=1}^n|A_{ij}|^{2-\epsilon}\right)^{1/(2-\epsilon)}\leq\max_{i\in[n]}\left(\sum_{j=1}^n|A_{ij}|^1\right)^{1/1}\,.
		\]
		Since $A$ is Hermitian, it follows that
		\[
		\left(\max_{i\in[n]}\left(\sum_{j=1}^n|A_{ij}|^2\right)^{1/2}\right)^2\leq
		\left(\max_{i\in[n]}\left(\sum_{j=1}^n|A_{ij}|^{2-\epsilon}\right)^{1/(2-\epsilon)}\right)^{2-\epsilon}\leq
		 \max_{i\in[n]}\sum_{j=1}^n|A_{ij}|\,,
		\]
		so the desired inequality follows.
	\end{proof}
\end{document}